\documentclass{article}
\usepackage{graphicx} 
\usepackage[utf8]{inputenc}
\usepackage{amsmath}
\usepackage{amsthm}
\usepackage{setspace}
\usepackage{bbm}
\usepackage{float}
\usepackage{subcaption}
\usepackage{algorithm}
\usepackage{algorithmic}
\usepackage[style=authoryear, hyperref=true, giveninits=true, uniquename=false, maxcitenames=2]{biblatex}
\usepackage[justification=centering]{caption}

\usepackage{array,tabularx,multirow,makecell}
\usepackage{tikz}
\usepackage[top=2cm,bottom=2cm,right=3cm,left=3cm]{geometry}
\usepackage{amssymb}

\theoremstyle{definition}

\newtheorem*{remark}{Remark}

\newtheorem{proposition}{Proposition}

\newcommand{\INPUT}{\item[\textbf{Input:}]}
\newcommand{\OUTPUT}{\item[\textbf{Output:}]}

\newcommand{\qbeta}[3]{Q_{\mathrm{Beta}(#1,#2)}\!\left(#3\right)}

\usepackage{xcolor}
\usepackage[colorlinks=true,
            citecolor=blue,
            linkcolor=blue,
            urlcolor=blue]{hyperref}

\AtBeginBibliography{%
}

\DeclareFieldFormat{citehyperref}{%
  \DeclareFieldAlias{bibhyperref}{noformat}%
  \bibhyperref{#1}}

\DeclareFieldFormat{textcitehyperref}{%
  \DeclareFieldAlias{bibhyperref}{noformat}%
  \bibhyperref{#1}}

\savebibmacro{cite}
\renewbibmacro*{cite}{%
  \printtext[citehyperref]{\restorebibmacro{cite}\usebibmacro{cite}}}

\savebibmacro{textcite}
\renewbibmacro*{textcite}{%
  \printtext[textcitehyperref]{\restorebibmacro{textcite}\usebibmacro{textcite}}}

\title{Beta-trees for testing multivariate goodness-of-fit and localizing deviations from a model 

}
\author{Valérie N. P. Ho\\Stanford University
 \and
 Guenther Walther\thanks{Research supported by NSF grant DMS-2413885.}  \\Stanford University
}
\date{}

\begin{document}

\maketitle

\begin{abstract}
   We introduce a novel goodness-of-fit (GOF) procedure based on Beta-tree partitions. A Beta-tree produces a data-adaptive partition of the sample space into regions  and provides guaranteed finite sample confidence intervals for the probability contents of each region. The proposed test assesses whether the probabilities assigned by a null distribution $F_0$ fall within these intervals, thereby quantifying agreement between the model and the data. A key application is the selection of the number of components in a mixture model, where the null distribution is constructed via $k$-means clustering. In contrast to classical global GOF tests such as Kolmogorov–Smirnov or Anderson–Darling, which quantify the discrepancy through a single global statistic, our method is designed to detect local departures from the null and to identify regions of model misspecification. We demonstrate the efficiency of our test in detecting departures from the null on some simulated and real datasets.
\end{abstract}
\section{Introduction}
Assessing the goodness-of-fit (GOF) of a statistical model is a fundamental problem in statistics, with wide-ranging applications in model selection, density estimation, and hypothesis testing. Given $X_1,\dots,X_n$ i.i.d. observations of a $d$-dimensional random variable $X$, GOF tests aim to evaluate whether this set of observed data is consistent with a specified null distribution $F_0$:
$$\mathcal{H}_0: X\sim F_0 $$
In the univariate setting, prominent examples include the Kolmogorov–Smirnov (\textcite{kolmogorov1933empirical}, \textcite{smirnov1948table}) and Anderson--Darling tests (\textcite{anderson1954test}), as well as the Cramér--von Mises (\textcite{cramer1928composition}, \textcite{vonmises1928wahrscheinlichkeit}), Berk--Jones (\textcite{berk1979goodness}), higher criticism tests (\textcite{donoho2004higher}) and tests based on phi-divergences (\textcite{jager2007goodness}), which rely on global metrics of discrepancy between the empirical distribution function and $F_0$. Such tests only indicate whether the model fits the data overall but provide limited insight into where the model fails.

In many practical scenarios, it is desirable to not only detect departures from the null but also to localize the regions of the sample space responsible for the deviation. Tests that provide such localized information can offer deeper insights into model misspecification, guide model refinement, and improve interpretability. This motivates the development of \textit{local} goodness-of-fit tests.  One example of such an approach via Bayesian inference is sketched in \textcite{ma2011coupling}, where Optional Pólya trees (OPT) are used to construct a partition of the sample space together with posterior probabilities for each rectangle in the partition, which describe the posterior probability that the data distribution coincides with the hypothesized distribution on that rectangle.  The authors report that the computational effort becomes prohibitive in higher dimensions (\textcite{wong2010optional}, \textcite{jiang2016computational}). It is also not clear how to handle the multiple inference involving the large number of posterior probabilities, and it appears that the hierarchical stopping in the OPT could potentially cause local deviations to be averaged out, masking fine-scale structure and resulting in premature termination of splits in regions with subtle variation.

In this paper, we introduce a new local goodness-of-fit test with frequentist finite sample guarantees. It is based on Beta-trees, a data-adaptive partition of the Euclidean space introduced in \textcite{walther2023beta}. This data-adaptive partition of the sample space into rectangles allows us to give exact finite sample confidence statements for the probability content of each rectangle, assuming only that the data distribution is continuous. Moreover, the Beta-tree allows for simultaneous inference across these rectangles in an optimal way.  Comparing these confidence intervals to the probabilities given by a candidate null distribution, we can construct a test for goodness-of-fit that is sensitive to local deviations. Beyond testing the null hypothesis, the method also provides a natural mechanism to localize deviations from the null, addressing a key limitation of classical global tests.

This paper is organized as follows: first, we give an overview of Beta-trees and their construction, followed by a description of the GOF test we propose. In Section \ref{sec:mixture}, we show how to employ this GOF test in order to select the number of components in a mixture model. In cases where we cannot analytically evaluate probabilities under the null, we propose in Section \ref{sec:GOF-MC} an approach based on Monte Carlo simulation combined with Clopper-Pearson confidence intervals that account for the uncertainty in the simulation. Finally, we illustrate in Section \ref{sec:examples} the relevance of our GOF test with some examples using simulated and real data. The code used for the numerical examples is available at \url{https://github.com/hoval58/beta-tree-gof}.

\section{Foundations of Beta-trees and formulation of a new GOF test}\label{sec:foundations}
\subsection{Building the $k$-d tree}\label{sec:building}
A Beta-tree is a statistical analogue of the classical $k$-d tree, a widely used data structure in computer science; see \textcite{bentley1975multidimensional,friedman1977algorithm}.
Its general construction is described in \textcite{walther2023beta}. Here we summarize the specific version used for our goodness-of-fit procedure, noting in particular that no subsequent pruning step of the $k$-d tree is required. In essence, a Beta-tree is based on a recursive partitioning of the space using the same splitting mechanism as in $k$-d trees. Thus, each rectangle is divided into two children rectangles according to the empirical distribution of the data, yielding a multiscale representation that adapts to local density features. Then the Beta-tree assigns finite sample guaranteed confidence intervals to each rectangle in the partition. 

Let $X_1,\dots,X_n \in \mathbb{R}^d$. The Beta-tree is built by recursively subdividing the sample space into axis-aligned regions, each subdivision determined by an empirical quantile along one coordinate direction. We begin with the initial region $R_0:=\mathbb{R}^d$. To create the first split, we choose a coordinate $p\in \{1,…,d\}$. We then cut along the $p$-th axis at the empirical median so that the two resulting regions contain approximately equal numbers of observations. Writing $X_{(k),p}$ for the $k$-th order statistic of $\{X_{1p},\dots,X_{np}\}$, the first two children of the root are 
$$R_1:=\{x\in R_0:x_p<X_{(\lceil \frac{n}{2}\rceil),p}\}, \quad R_2:=\{x\in R_0:x_p>X_{(\lceil \frac{n}{2}\rceil),p}\}$$
Then, the same scheme is applied recursively within each resulting region. For a given node $R_k$, we determine the next splitting coordinate by cycling through the axes: denoting the depth of $R_k$ as $D:=\lfloor \log_2 (k+1)\rfloor$, then the coordinate used at $R_k$ is $p=D\mod d+1$.
To form the children of $R_k$, we restrict attention to the observations lying in $R_k$ and cut along coordinate $p$ at their empirical median. This produces two subregions $R_{2k+1}$ and $R_{2k+2}$, each containing about half of those observations (the points on the split boundaries are not themselves included in either child region). The recursion continues in this manner until a region contains fewer than $4\log n$ observations, at which point it is no longer subdivided.

\subsection{Constructing confidence intervals for the regions in the $k$-d tree}
Let $F$ denote the distribution of the $X_i$. The Beta-tree assigns simultaneous confidence intervals for the probability content of each random rectangle $R_k$ of the partition, relying on the crucial fact that $F(R_k)$ follows a known Beta distribution, see Proposition $1$ of \textcite{walther2023beta}. Strictly speaking, because $R_k$ is itself random, these intervals are prediction intervals for the random quantity $F(R_k)$, but we follow \textcite{walther2023beta} in referring to them as confidence intervals. An exact finite sample $1-\alpha$ level confidence interval for $F(R_k)$ is given by:

\begin{equation}\label{eq:conf_int}
C_k(\alpha) :=
\Bigl(
\qbeta{n_k+1}{\,n-n_k}{\alpha/2},\;
\qbeta{n_k+1}{\,n-n_k}{1-\alpha/2}
\Bigr).
\end{equation}

where $\qbeta{a}{\,b}{t}$ denotes the $t$-quantile of the Beta distribution with shape parameters $a$ and $b$, and $n_k$ is the (deterministic) number of observations in $R_k$.

In order to make those confidence intervals simultaneously valid for all rectangles in the Beta-tree, we apply the weighted Bonferroni correction suggested in \textcite{walther2023beta}: for each of the $N_D$ rectangles at tree depth $D\geq 1$, we construct a confidence interval for the probability content of such a rectangle using significance level
\begin{equation}\label{eq:bonferroni}
    \alpha_D=\frac{\alpha}{N_D(D_{\max}-D+2)\sum_{B=2}^{D_{\max}+1}\frac{1}{B}}
\end{equation}

where $\alpha$ denotes the global significance level we want to achieve for the whole $k$-d tree and $D_{\max}$ denotes the maximum depth of the Beta-tree. 

Thus, using $\alpha_D$ in (\ref{eq:conf_int}) for every rectangle $R_k$ at depth $D$ gives simultaneous coverage $1-\alpha$ for the $F(R_k)$. 

Note one difference in the definition of $\alpha_D$ compared to \textcite{walther2023beta}: since we do not restrict attention to bounded rectangles, we consider in \eqref{eq:bonferroni} the number $N_D$ of all rectangles at depth $D$ in the constructed tree (clearly $N_D\leq 2^D$, but equality need not hold since the rectangles at the same depth may not contain exactly the same number of observations and therefore splitting may terminate at slightly different depths).

\subsection{The goodness-of-fit test}\label{sec:gof}
Having introduced the Beta-tree along with the confidence intervals for the probability mass in each of its regions, we now describe how these components are combined to form our goodness-of-fit test. In contrast to the usual construction of the Beta-tree histogram, we do not apply the pruning step that merges regions whose data appear uniform. For the purpose of our test, we retain the full collection of nodes in the Beta-tree, which provides a finer partition and enables detection of a wider range of departures from the null model.

We begin by evaluating the probability assigned to each region under the null distribution, assuming for now that we are in a setting where these probabilities can be computed exactly. Then the GOF test checks for each rectangle whether the null probability content falls in the corresponding confidence interval. This will define a \textit{score} of the GOF test. In more detail, given a dataset $\mathcal{X}$ and a corresponding Beta-tree defining $N$ rectangles $\{R_k\}_{k=1}^N$, each associated with a confidence interval $C_k:=C_k(\alpha_k)$ of level $\alpha_k$ (as defined in \eqref{eq:bonferroni}) for its probability content, the \textit{score} of the GOF test is given by
$$\phi(\mathcal{X}):=\frac{1}{N}\sum_{k=1}^N \mathbbm{1}_{\{F_0(R_k)\in C_k\}}$$
 where $F_0(R_k)$ denotes the probability content of $R_k$ under the null.
Now the outcome of our GOF test is as follows:
\begin{itemize}
    \item If the score equals $1$, we do not reject the null hypothesis: in this case, the null distribution is consistent with the observed data.
    \item If the score is strictly less than $1$, at least one region $R_k$ is significant, i.e. $F_0(R_k)$ is not contained in the confidence interval $C_k$. We therefore reject the null hypothesis and report the rectangles that are significant, as detailed below.
\end{itemize}
An important advantage of our test is that it can pinpoint the specific regions of the data where the null distribution fails to fit. Thus our GOF test not only allows rejection of the null hypothesis with a finite sample guaranteed significance level $\alpha$, but it also provides a \textit{localization} of the areas in which the observed data deviate from the null distribution: due to the construction of our procedure, we can claim with finite sample confidence level $1-\alpha$ that the hypothesized null distribution is violated on every rectangle that is flagged as significant. That is, this claim holds with simultaneous confidence $1-\alpha$ across the flagged rectangles.

Instead of reporting all significant rectangles, our procedure reports only the minimal (with respect to inclusion) rectangles among the significant rectangles, i.e. the significant rectangles which do not strictly contain another significant rectangle. This is motivated by the fact that a local violation of the null hypothesis may not only result in a corresponding small rectangle $R_k$ to be significant, but it may also result in rectangles that are supersets of $R_k$ to be flagged as significant. In this case one would only be interested in the smaller $R_k$ since a smaller rectangle allows for a better localization.

\section{Application to Gaussian mixture model selection}\label{sec:mixture}

A natural application of our GOF test is in selecting the number of components in a mixture model. In many practical situations, data are believed to arise from a mixture of distributions, but the true number of components is unknown. Our test can be used to evaluate how well candidate mixture models fit the observed data, providing a principled way to choose the appropriate number of components. We focus our application on Gaussian mixture models.

Suppose the data arise from a mixture of $k$ Gaussian mixture components, where $k$ is unknown and must be selected. A natural approach is to consider a range of candidate values $j$ and, for each $j$, fit a $j$-component mixture using the $k$-means algorithm. This produces estimates of the mixture weights, component means, and covariance matrices, which we can use to define the corresponding null distribution. Given these parameters, the probability content of any region of the Beta-tree under the $j$-component mixture can be computed explicitly. This allows us to evaluate the null mixture distribution on every region of the tree and to apply our goodness-of-fit test for each candidate $j$. By comparing the resulting scores across values of $j$, we can identify the number of mixture components that is most consistent with the observed data. 

More precisely, we select the number of mixture components $j^{*}$ as the smallest candidate whose corresponding null distribution achieves a GOF score of $1$. If no candidate achieves score $1$, we instead select the candidate whose null distribution attains the highest GOF score.
This flexibility accounts for the inherent variability introduced by parameter estimation via the $k$-means algorithm, which may prevent the null distribution from perfectly matching the observed data. By prioritizing the model with the highest score rather than requiring a perfect score, our procedure balances goodness-of-fit with the uncertainty in parameter estimation. The pseudocode of this procedure is given in Algorithm \ref{alg:gof-mixture}.

\begin{algorithm}[t]
\caption{Beta-tree--based GOF test for mixture model selection\label{alg:gof-mixture}}
\begin{algorithmic}[1]
\INPUT 
    An increasing sequence $J=\{j_1,\dots,j_\ell\}$ of candidate numbers of mixture components

\OUTPUT The selected number of mixture components $j^{*}$

\FOR{$j \in J$}
    \STATE Apply the $k$-means algorithm to partition the data into $j$ clusters and estimate mixture parameters 
    $\{(\hat{w}_i,\hat{\mu}_i,\hat{\Sigma}_i)\}_{i=1}^j$
    \STATE Define the null distribution
    \[
    F_{0,j}(x) := \sum_{i=1}^{j} \hat{w}_i \, F(x;\hat{\mu}_i,\hat{\Sigma}_i),
    \]
    where $F(\,\cdot\,;\mu,\Sigma)$ denotes the distribution function of $\mathcal{N}(\mu,\Sigma)$.
    \STATE Apply the Beta-tree GOF test described in Section~2.3 using $F_{0,j}$ as the null distribution and record the resulting score $s_j$
    \IF{$s_j = 1$}
        \RETURN $j^{*} := j$
    \ENDIF
\ENDFOR
\RETURN $j^{*} := \arg\max_{j \in J} s_j$
\end{algorithmic}
\end{algorithm}

\section{GOF Testing with Monte Carlo Simulation of Null Probabilities}
\label{sec:GOF-MC}

In some settings, the null probabilities associated with the regions of the Beta-tree are not available in closed form and must be approximated. A natural approach is to use Monte Carlo simulation: one draws a large independent sample from the null distribution and approximates each region’s null probability by the empirical proportion of simulated points falling into that region. This procedure yields unbiased estimators of the null probabilities. In the following we describe a version of the GOF test which accounts for the uncertainty arising from such a Monte Carlo sampling. 

 In order to approximate the hypothesized null distribution, we simulate a large sample  of size $m$ (with $m= Cn$ and $C\gg 1$), assuming we have a sampler to do this. The number of MC observations that fall in a fixed rectangle $R$ follows a binomial distribution with parameters $m$ and $F_0(R)$. We use the Clopper-Pearson confidence interval for $F_0(R)$, see \textcite{clopper1934use} and \textcite{Thulin_2014}. This interval, which uses quantiles of the Beta distribution, is the classical “exact” binomial confidence interval obtained by inverting the corresponding binomial test and it guarantees at least the nominal coverage for finite $m$. A non-empty intersection between this interval and the corresponding confidence interval given by the Beta-tree
indicates that the data are consistent with the hypothesized null distribution on that rectangle. More formally, the score is now  defined by 
$$\phi(\mathcal{X}):=\frac{1}{N}\sum_{k=1}^N \mathbbm{1}_{\{D_k \cap C_k\neq \emptyset\}}$$
where $C_k(\alpha_k)$ and $D_k(\tilde{\alpha}_k)$ denote the Beta-tree and Clopper-Pearson confidence intervals of levels $1-\alpha_k$ and $1-\tilde{\alpha}_k$ for $F_0(R_k)$, respectively. In practice, we will set $\alpha_k=\tilde{\alpha}_k$. We can make the following statement about the type-I error of our test:

\begin{proposition}
Suppose that for each $k=1,\dots,N$, $\alpha_k=\tilde{\alpha}_k$ and $\sum_{k=1}^N\alpha_k=\frac{\alpha}{2}$. Under the null hypothesis, conditional on the Beta-tree regions $\{R_k\}_{k=1}^N$:
\[
\mathbb{P}_0\!\left(\phi(\mathcal X)<1 \,\middle|\, \{R_k\}_{k=1}^N\right)
\le \alpha
\]
And it follows that
\[
\mathbb{P}_0\left(\phi(\mathcal X)<1\right)\le \alpha
\]
\end{proposition}

\begin{proof}
Note that $\phi(\mathcal X)<1$ if and only if there exists $k\in\{1,\dots,N\}$ such that
$C_k(\alpha_k)\cap D_k(\tilde\alpha_k)=\emptyset$. Hence,
\[
\{\phi(\mathcal X)<1\}
= \bigcup_{k=1}^N \{C_k(\alpha_k)\cap D_k(\tilde\alpha_k)=\emptyset\}.
\]
Fix $k$. Under $H_0$, both $C_k(\alpha_k)$ and $D_k(\tilde\alpha_k)$ are confidence intervals
for the same quantity $F_0(R_k)$ (conditional on $\{R_k\}_{k=1}^N$). Therefore,
\begin{align*}
\mathbb{P}_0\!\left(C_k(\alpha_k)\cap D_k(\tilde\alpha_k)=\emptyset \,\middle|\, \{R_j\}_{j=1}^N\right)
&\le 
\mathbb{P}_0\left(F_0(R_k)\notin C_k(\alpha_k)\,\middle|\, \{R_j\}_{j=1}^N\right)
+
\mathbb{P}_0\left(F_0(R_k)\notin D_k(\tilde\alpha_k)\,\middle|\, \{R_j\}_{j=1}^N\right) \\
&\le \alpha_k+\tilde\alpha_k.
\end{align*}
Applying a union bound yields
$$
\mathbb{P}_0\left(\phi(\mathcal X)<1 \,\middle|\, \{R_k\}_{k=1}^N\right)
\le \sum_{k=1}^N
\mathbb{P}_0\left(C_k(\alpha_k)\cap D_k(\tilde\alpha_k)=\emptyset \,\middle|\, \{R_j\}_{j=1}^N\right)
\le \sum_{k=1}^N(\alpha_k+\tilde\alpha_k)=\frac{\alpha}{2}+\frac{\alpha}{2}=\alpha
$$
And by the tower property, it also follows that
$\mathbb{P}_0\!\left(\phi(\mathcal X)<1\right)\le \alpha$.
\end{proof}
\begin{remark}
Another possible GOF procedure, when the null region probabilities are not available in closed form, is to replace $F_0(R)$ by its Monte Carlo estimate $\hat{F}_0(R)$ and then apply the standard test from Section \ref{sec:gof}, i.e., check whether $\hat{F}_0(R)$ falls within the Beta-tree confidence intervals. In simulations, this plug-in approach tended to yield higher power, but it comes at the cost of losing theoretical type-I error control, unlike the test proposed above.
\end{remark}

\section{Simulation results and applications to real data}\label{sec:examples}

\subsection{A toy example}

We illustrate our GOF test with a simple toy example. We sample $1500$ data points  from the bivariate distribution $\mathcal{N}\left(\begin{pmatrix}
    5\\
    5
\end{pmatrix},\begin{pmatrix}
    1 &0.5\\
    0.5&1
\end{pmatrix}\right)$. 

The Beta-tree for one such sample is shown in Figure \ref{fig:beta_tree_plot}. The shaded rectangles are the leaf cells of the  Beta-tree as defined in \textcite{walther2023beta}. The GOF test, however, considers all rectangles created during the recursive partitioning, including internal (parent) rectangles, shown by the split lines.

\begin{figure}
    \centering
    \includegraphics[width=0.7\linewidth]{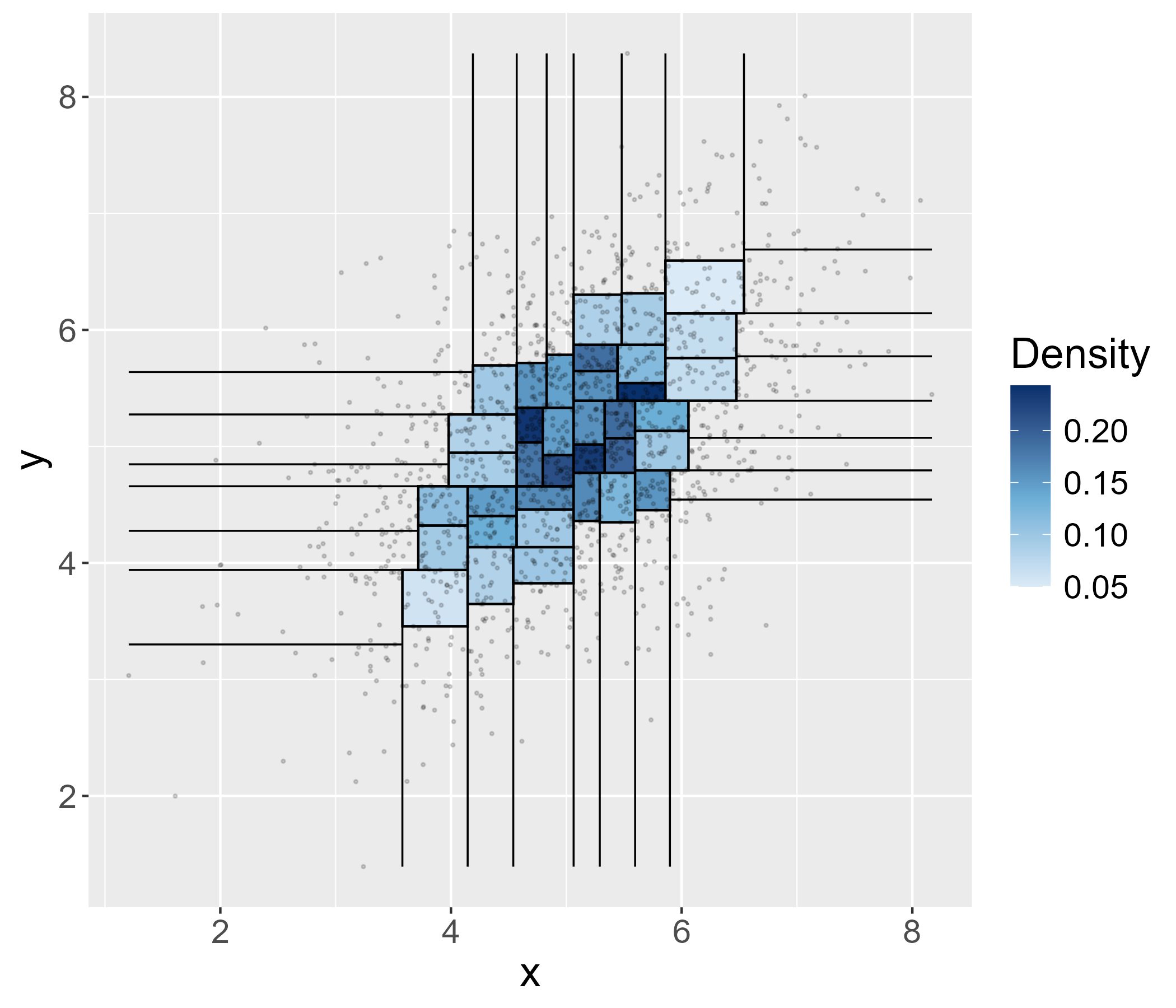}
    \caption{Beta-tree histogram for $1500$ samples from the bivariate normal $\mathcal{N}\left(\begin{pmatrix}
    5\\
    5
\end{pmatrix},\begin{pmatrix}
    1 &0.5\\
    0.5&1
\end{pmatrix}\right)$. }
    \label{fig:beta_tree_plot}
\end{figure}

If we did not know the true distribution and wanted to test whether our data were generated from a shifted bivariate normal distribution $\mathcal{N}\left(\begin{pmatrix}
    5\\
    5.15
\end{pmatrix},\begin{pmatrix}
    1 &0.5\\
    0.5&1
\end{pmatrix}\right)$, we would first construct a Beta-tree as described in section \ref{sec:building}, and then run the GOF test presented in section \ref{sec:gof}, taking this shifted distribution as the null. 

In order to estimate the power of our test, we draw $1000$ samples of size $1500$ from the true distribution and compute the proportion of replicates for which the test detects a deviation from the putative null hypothesis (the shifted distribution).  We obtain a power of $0.94$. 
We show in Figure~\ref{fig:beta-tree-flagged} the minimal (with respect to set inclusion) rectangles flagged by our test. The boundaries of these two rectangles are marked in red. In each flagged rectangle $R$, we are $1-\alpha$ confident (in the simultaneous sense) that $F_0(R)$ is incompatible with the data. This illustrates that the proposed GOF test not only detects departures from the null, but also localizes them by identifying the specific regions of the sample space where, with overall confidence level $1-\alpha$, the discrepancies occur.

We also estimate the type I error rate using $1000$ replicates, running the GOF test with the true distribution as the null. The resulting type I error rate we get is $0.07$. The Beta-tree was constructed with $\alpha=0.1$.

\begin{figure}
    \centering
    \includegraphics[width=0.7\linewidth]{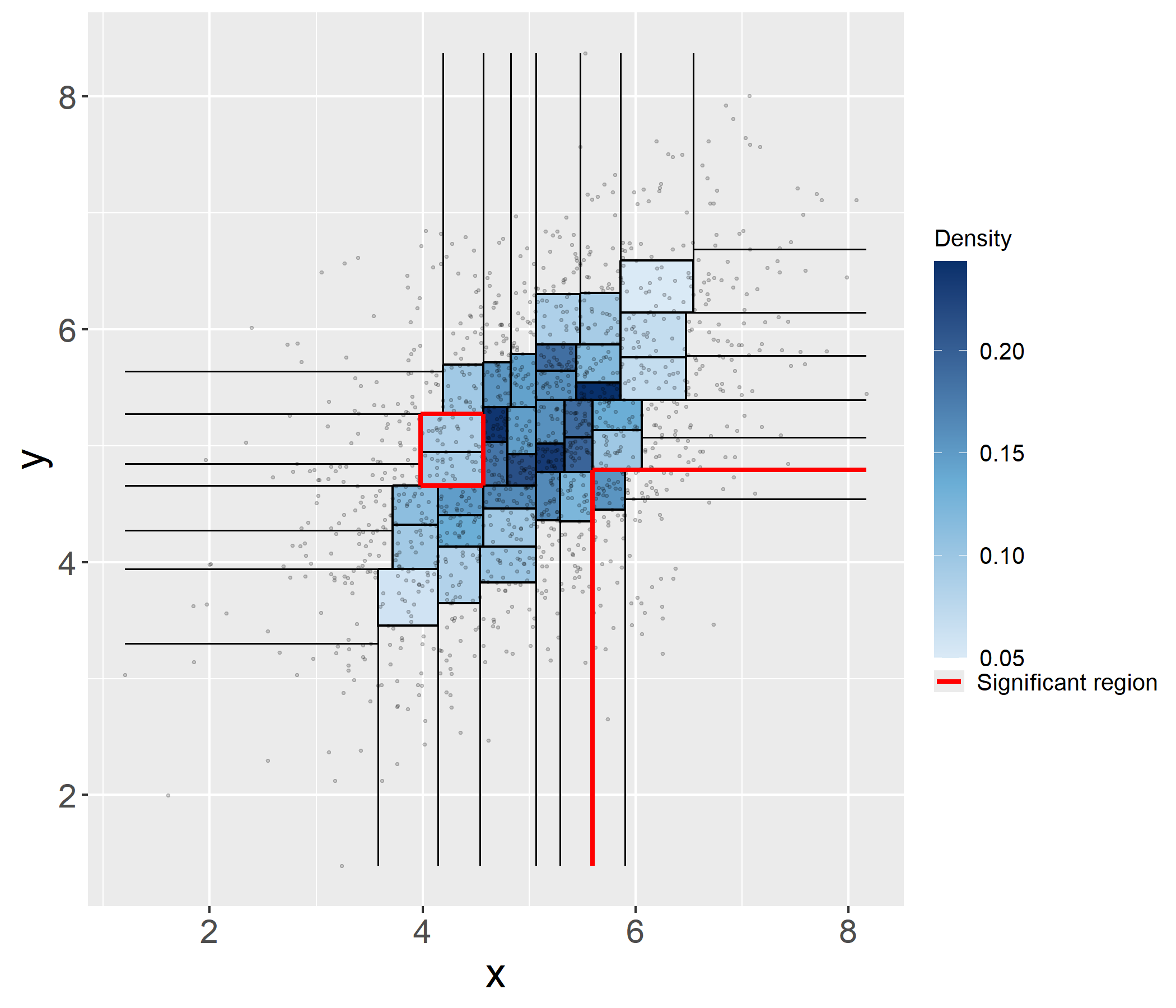}
    \caption{Beta-tree histogram for $1500$ samples from the bivariate normal $\mathcal{N}\left(\begin{pmatrix}
    5\\
    5
\end{pmatrix},\begin{pmatrix}
    1 &0.5\\
    0.5&1
\end{pmatrix}\right)$ with regions flagged by the GOF test when the putative null distribution is  $\mathcal{N}\left(\begin{pmatrix}
    5\\
    5.15
\end{pmatrix},\begin{pmatrix}
    1 &0.5\\
    0.5&1
\end{pmatrix}\right)$  }
    \label{fig:beta-tree-flagged}
\end{figure}

\subsection{Mixture model selection example}
We now present one example where our GOF test is applied to Gaussian mixture model selection, as described in section~\ref{sec:mixture}. First, we consider an `easy' setting in which we generate a sample of size $5000$ from a three-dimensional Gaussian mixture model ($P_{\textrm{well-sep}}$) in which the components are well-separated. The second setting we consider has $10000$ samples from a two-dimensional Gaussian mixture ($P_{\textrm{overlap}}$) for which the mixture components exhibit substantial overlap. We use the following parameters:
\begin{align*}
    P_{\textrm{well-sep}}&=0.25\hspace{0.2 em}\mathcal{N}\left(\begin{pmatrix}
    -1.5\\
    0.6\\
    1
\end{pmatrix},\begin{pmatrix}
    1 &0.5&0.5\\
    0.5&1&0.5\\
    0.5&0.5&1\\
\end{pmatrix}\right)+0.5\hspace{0.2 em}\mathcal{N}\left(\begin{pmatrix}
    2\\
    -1.5\\
    0
\end{pmatrix},\begin{pmatrix}
    1 &0&0\\
    0&1&0\\
    0&0&1\\
\end{pmatrix}\right)\\&+0.25\hspace{0.2 em}
\mathcal{N}\left(\begin{pmatrix}
    -2.6\\
    -3\\
    -2
\end{pmatrix},\begin{pmatrix}
    1 &-0.4&0.6\\
    -0.4&1&0\\
    0.6&0&1\\
\end{pmatrix}\right)
\end{align*}
$$ P_{\textrm{overlap}}=0.4 \hspace{0.2em}\mathcal{N}\left(\begin{pmatrix}
    0\\
    0
\end{pmatrix},\begin{pmatrix}
    1.4 &0.8\\
    0.8&1.2
\end{pmatrix}\right)+0.35\hspace{0.2 em}\mathcal{N}\left(\begin{pmatrix}
    2.2\\
    1.8
\end{pmatrix},\begin{pmatrix}
    1.2 &-0.7\\
    -0.7&1.3\\
\end{pmatrix}\right)+0.25 \hspace{0.2 em}
\mathcal{N}\left(\begin{pmatrix}
    3.2\\
    -0.8
\end{pmatrix},\begin{pmatrix}
    0.9 &0.4\\
    0.4&1.1
\end{pmatrix}\right)$$

We determine the `best' $k$, returned by our GOF procedure described in Algorithm~\ref{alg:gof-mixture}, which is the number $k$ that achieves the maximal score (not necessarily equal to $1$) across all the values of $k$ considered. Running our GOF test over $1000$ different seeds, we find that in the well-separated and overlap cases, our test recovers the true number of mixture components $97.2\%$ and $98.8\%$ of the time, respectively. For some replicates in both settings, the maximal GOF score is strictly less than 1, yet it is still attained at the true number of mixture components in most cases.

We show in Figure \ref{fig:overlap} some data sampled from $P_{\textrm{overlap}}$ with the corresponding Beta-tree histogram along with flagged rectangles, for $k=3$ and $k=2$. 
For $k=3$, the few significant rectangles are concentrated primarily in the regions where the mixture components overlap. This suggests that the fitted three-component model captures the global cluster structure reasonably well, and that the residual discrepancies are confined to the boundary regions where the components are hardest to distinguish. For $k=2$, we observe many flagged rectangles spread across the space, showing clear underfitting.

\begin{figure}
    \centering

    \begin{subfigure}{0.49\textwidth}
        \centering
        \includegraphics[width=\linewidth]{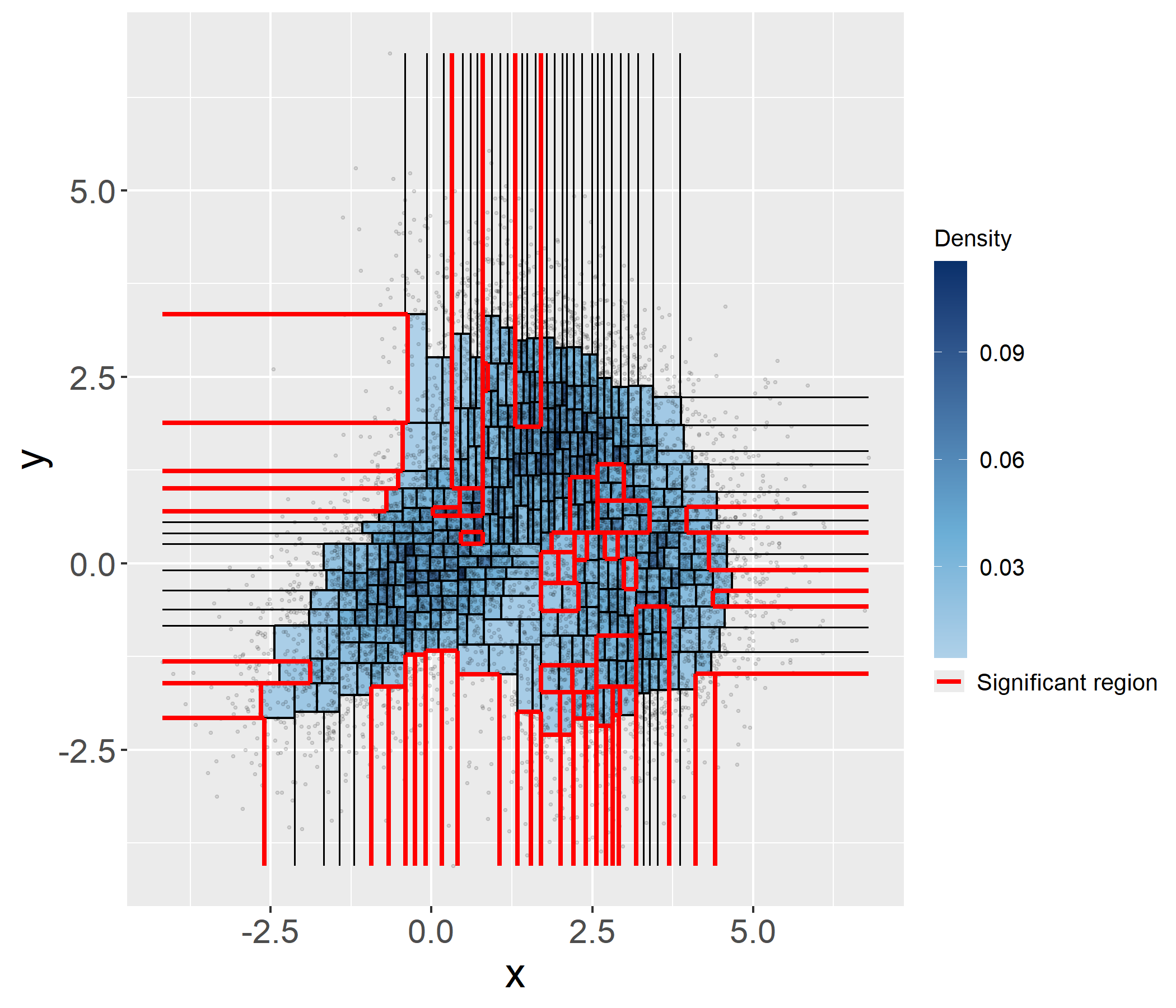}
        \caption{$k=2$}
    \end{subfigure}
    \hfill
    \begin{subfigure}{0.49\textwidth}
        \centering
        \includegraphics[width=\linewidth]{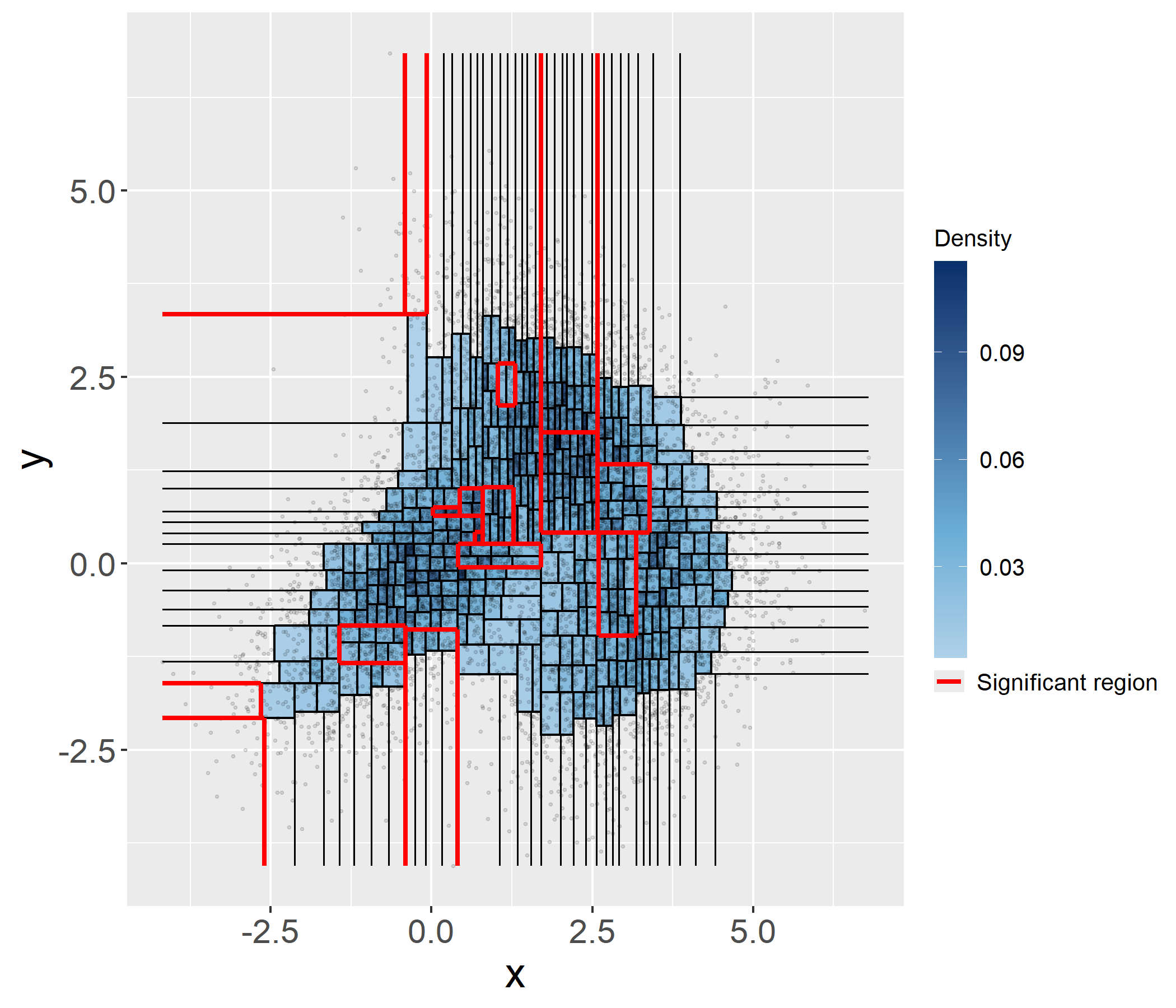}
        \caption{$k=3$}
    \end{subfigure}

    \caption{Beta-tree histogram for the $3$-component mixture of bivariate normals with substantial overlap, with flagged regions shown when running the Beta-tree GOF test procedure in combination with the $k$-means algorithm, with $k=2$ and $k=3$. Even for the correct number of components $k=3$, some regions are flagged due to imperfect parameter estimation from $k$-means, especially in the overlapping regions.}
    \label{fig:overlap}
\end{figure}

\subsection{When $F_0(R_k)$ is intractable: Monte Carlo-based GOF}
In this example, we consider the case where $F_0(R_k)$ is intractable and Monte Carlo simulation appears as a natural way to approximate the probability content under the null distribution, as described in Section \ref{sec:GOF-MC}.
We sample bivariate Beta data using a Gaussian copula with observation-specific random correlation. More precisely, samples are obtained through the following procedure:
\begin{enumerate}
    \item Sample $\eta\sim \mathcal{N}(\mu,\sigma^2)$. Set $\rho=\tanh(\eta)$
    \item Sample $(z_1,z_2)\sim \mathcal{N}\left(\begin{pmatrix}
    0\\
    0
\end{pmatrix},\begin{pmatrix}
    1 &\rho\\
    \rho&1
\end{pmatrix}\right)$
    \item Transform $(z_1,z_2)$ to be uniformly distributed: $u_1=\Phi(z_1)$, $u_2=\Phi(z_2)$ where $\Phi$ denotes the cdf of $\mathcal{N}(0,1)$.
    \item Transform to Beta marginals: $(x_1,x_2)=(F^{-1}_{(\alpha_1,\beta_1)}(u_1),F^{-1}_{(\alpha_2,\beta_2)}(u_2))$ where $F_{(\alpha,\beta)}$ denotes the cdf of the Beta distribution with shape parameters $\alpha$, $\beta$.
\end{enumerate}
Then, for any rectangle $R$, we have that $F_0(R)=\mathbb{E}_{\rho}[F_0(R|\rho)]$, which does not have a simple closed form expression. 

Thus, as described in section \ref{sec:GOF-MC}, we use Monte Carlo samples to form a Clopper-Pearson confidence interval, and use the GOF test based on the intersection between that interval and the confidence interval given by the Beta-tree.

In our simulations, we sample our data using $(\mu,\sigma,\alpha_1,\beta_1,\alpha_2,\beta_2)=(0,0.8,5,10,1,3)$. Suppose that we have imperfect knowledge of that model, and want to test if the data was generated using $(\mu,\sigma,\alpha_1,\beta_1,\alpha_2,\beta_2)=(0,0.8,5,10,1,3.8)$ instead. Using $1000$ independent datasets of size $n=2000$, we run the proposed MC-GOF test with Monte Carlo sample size $m=10000$ to estimate power and type-I error. For this parameter mismatch, the test rejects the null with empirical power $0.94$, while the empirical type-I error under the true null is $0.001$.

We note that this procedure is quite conservative, which is expected since we control errors simultaneously across many regions and use Clopper–Pearson Monte Carlo intervals, both of which tend to widen the acceptance region under the null in finite samples.

\subsection{Real data example: Insurance data} 

We run our Beta-tree GOF test on the \verb|loss| insurance data from the R \verb|copula| package, and focus on the dependence relationship between indemnity payment (\verb|loss|, which we denote as $X_1$) and allocated loss adjustment expense (\verb|alae|, which we denote as $X_2$). Following \textcite{genest2006goodness}, we restrict attention to the $n=1466$ uncensored observations and work with the associated pseudo-observations $U_{ij}=\frac{R_{ij}}{n+1}$ where $R_{ij}$ is the rank of $X_{ij}$ among $X_{1j},\dots, X_{nj}$, for $j\in\{1,2\}$. These data were used by various authors (e.g. \textcite{frees1998understanding}, \textcite{klugman1999fitting}) to perform copula-based model selection. 

In the bivariate setting, we recall that a copula is a distribution function \(C:[0,1]^2\to[0,1]\) with uniform margins, used to describe the dependence structure between two continuous random variables $X_1$ and $X_2$ separately from their marginal distributions. By Sklar's theorem, if $F$ and $G$ denote the marginal distribution functions of $X_1$ and $X_2$, then their joint distribution function can be written as
\[
H(x_1,x_2)=C(F(x_1),G(x_2)), \qquad x_1,x_2\in\mathbb{R}.
\]
A widely used parametric subclass is the family of Archimedean copulas, which in the bivariate case take the form
\[
C_\theta(u,v)=\phi_\theta^{-1}\bigl(\phi_\theta(u)+\phi_\theta(v)\bigr), \qquad (u,v)\in[0,1]^2,
\]
where \(\phi_\theta\) is a suitable generator depending on a parameter $\theta$. 

As in \textcite{genest2006goodness}, we consider three bivariate Archimedean copula families, namely the Clayton, Frank, and Gumbel--Hougaard copulas, whose bivariate distribution functions are given by
\[
C_\theta^{\mathrm{Clayton}}(u,v)
=
\left(u^{-\theta}+v^{-\theta}-1\right)^{-1/\theta},
\qquad \theta>0,
\]

\[
C_\theta^{\mathrm{Frank}}(u,v)
=
-\frac{1}{\theta}
\log\left(
1+\frac{(e^{-\theta u}-1)(e^{-\theta v}-1)}{e^{-\theta}-1}
\right),
\qquad \theta\in\mathbb{R}\setminus\{0\},
\]

\[
C_\theta^{\mathrm{Gumbel}}(u,v)
=
\exp\left(
-\left((-\log u)^\theta+(-\log v)^\theta\right)^{1/\theta}
\right),
\qquad \theta\ge 1.
\]
These families represent different forms of dependence: the Clayton copula exhibits lower-tail dependence, the Gumbel--Hougaard copula exhibits upper-tail dependence, and the Frank copula is symmetric and has no tail dependence.

The goodness-of-fit procedures of \textcite{genest2006goodness} are based on Kendall's process, namely $\sqrt{n}(K_n-K(\theta_n,\cdot))$, where $K_n$ denotes the empirical distribution function of $H(X,Y)$ and its parametric counterpart $K(\theta_n,\cdot)$ under the fitted copula model. From this process, they consider two test statistics,
\[
S_n=\int_0^1 \mathbb{K}_n(t)^2\, k(\theta_n,t)\,dt
\qquad\text{and}\qquad
T_n=\sup_{0\le t\le 1} |\mathbb{K}_n(t)|,
\]
where \(\mathbb{K}_n(t)=\sqrt{n}\{K_n(t)-K(\theta_n,t)\}\) and \(k(\theta_n,t)\) is the density of \(K(\theta_n,t)\). Thus $S_n$ is a Cram\'er--von Mises type statistic, while $T_n$ is a Kolmogorov--Smirnov type statistic. They also report \(S_n^{(0)}\), which corresponds to the Wang--Wells (\textcite{wang2000model}) $L^2$-type statistic with truncation parameter equal to $0$, i.e. $S_{n}^{(0)}=\int_{0}^1\mathbb{K}_n(t)^2dt$.

For the LOSS/ALAE data, Genest et al.\ estimate the dependence parameter in each family by inversion of Kendall's tau and compute bootstrap $p$-values for \(S_n\), \(T_n\), and \(S_n^{(0)}\). Their results reject the Clayton and Frank copulas at the \(5\%\) level, whereas the Gumbel--Hougaard copula is not rejected. Indeed, the reported $p$-values for the Gumbel--Hougaard model are all above 80\% (\textcite[Table~7]{genest2006goodness}).

Our Beta-tree GOF analysis leads to the same qualitative conclusion. In this example, the null probability assigned to each Beta-tree rectangle is computed analytically under the fitted copula model, so no Monte Carlo simulation is needed. When applied to the pseudo-observations from the uncensored LOSS/ALAE sample, the only one of the three copula models that passes our GOF test is the Gumbel--Hougaard copula (score $=1$), while the Clayton and Frank copulas are flagged as inconsistent with the data (score $<1$). More importantly, unlike the global statistics \(S_n\), \(T_n\), and \(S_n^{(0)}\), our method identifies the minimal significant rectangles that provide a localized diagnostic of where the fitted copula model appears inconsistent with the data (see Figure~\ref{fig:clayton} for the Clayton copula). Note that since the pseudo-observations are rank-based and hence discrete, the exact finite-sample guarantees established above for continuous data do not directly apply here, although the resulting rectangles still provide an informative localized diagnostic.

\begin{table}[ht]
\centering
\caption{Results of the Beta-tree GOF test performed on the insurance data, with Gumbel, Frank and Clayton copulas, using the values of $\theta_n$ reported in \textcite{genest2006goodness} as estimates of the dependence parameter $\theta$, with $\alpha=0.1$.}
\label{tab:type1-power}
\begin{tabular}{|c|c|c|}
\hline
Copula model & $\theta_n$ & Score \\
\hline
Gumbel--Hougaard & $1.468$ & $1.0$ \\
\hline
Frank & $3.143$ & $0.98$ \\
\hline
Clayton & $0.939$ & $0.92$\\
\hline
\end{tabular}
\end{table}

\begin{figure}
    \centering
    \includegraphics[width=0.7\linewidth]{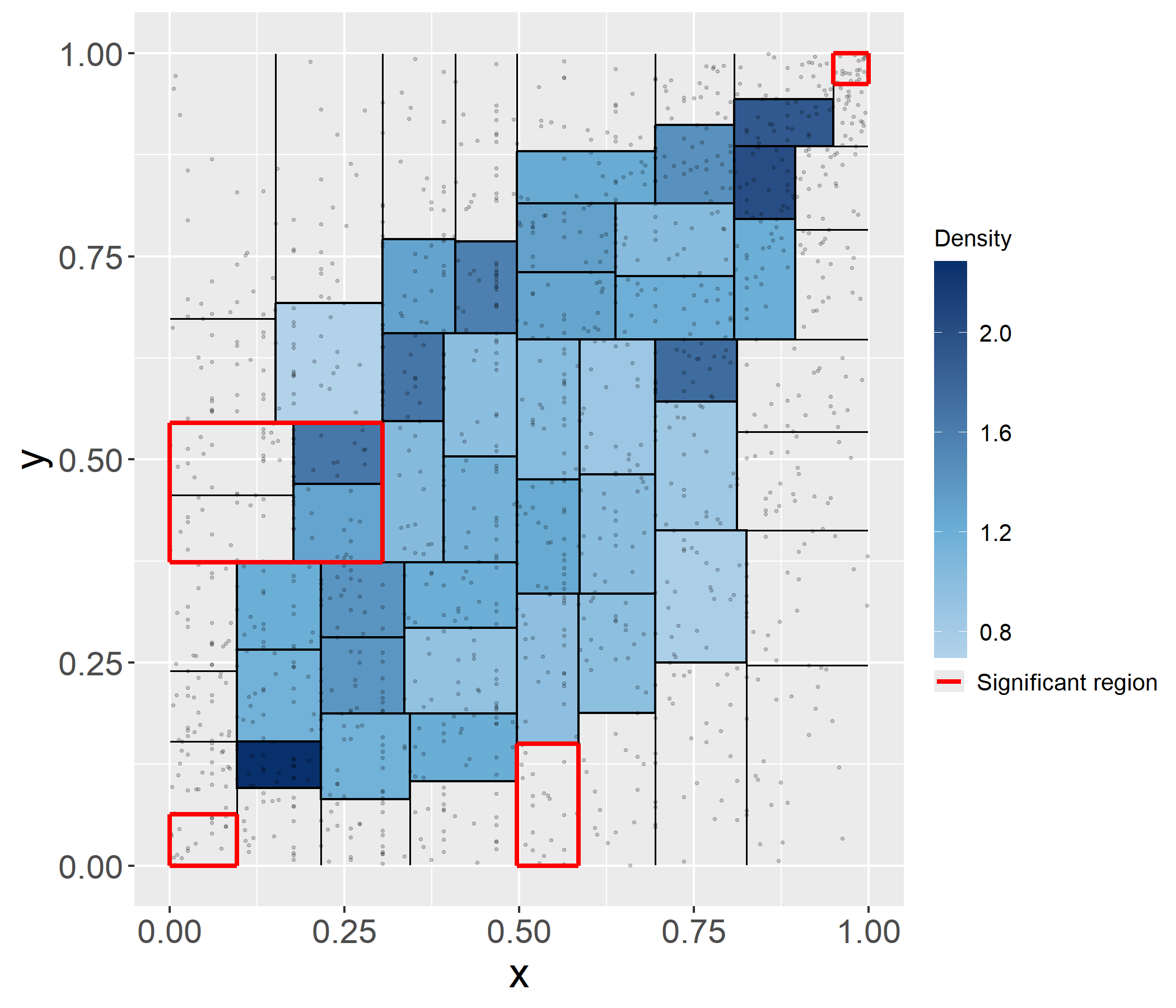}
    \caption{Beta-tree histogram of the pseudo-observations obtained for the \texttt{loss} and \texttt{alae} variables of the \texttt{loss} dataset (shaded regions are the leaves of the Beta-tree), with regions flagged by our GOF test when taking the null distribution to be the Clayton-copula with $\theta_n=0.939$}
    \label{fig:clayton}
\end{figure}

\section{Conclusion}

We have proposed a new goodness-of-fit procedure based on Beta-tree partitions. 
The main idea is to use the data-adaptive rectangles generated by the Beta-tree together with simultaneous confidence intervals for their probability contents. 
By comparing these intervals with the probabilities assigned by a candidate null distribution, the resulting test provides a finite-sample, distribution-free way to assess whether the null model is compatible with the data. 
One advantage of this method is that it not only returns a global decision, but also localizes the regions of the sample space where the null model appears to fail with high probability.

The examples considered in this paper illustrate several aspects of the procedure. 
In simulated settings, the test is able to detect moderate departures from the null while localizing the discrepancies via the minimal significant rectangles. 
In the mixture model example, the Beta-tree score gives a natural diagnostic for comparing fitted Gaussian mixture models with different numbers of components. 
In particular, even when the fitted model does not achieve a perfect score because of parameter estimation error, the score can still identify the correct model complexity. The flagged regions also provide useful information about the nature of the misspecification, for instance by showing whether discrepancies are spread throughout the sample space or concentrated in overlapping cluster regions.

We also considered the case where the null probabilities of the Beta-tree rectangles are not available in closed form. 
For this setting, we introduced a Monte Carlo version of the test based on Clopper--Pearson intervals for the simulated null probabilities. 
This modification preserves finite-sample type-I error control by accounting for both the uncertainty in the data-driven Beta-tree prediction intervals and the uncertainty from Monte Carlo approximation. As expected, this version can however be conservative in finite samples, since it combines simultaneous inference over many regions with exact binomial confidence intervals.

Finally, the insurance data example shows how the proposed method can be used as a copula model assessment. In agreement with previous goodness-of-fit analyses based on Kendall's process, the Gumbel--Hougaard copula is not rejected, while the Clayton and Frank copulas are found to be inconsistent with the data. Beyond this global conclusion, the Beta-tree test identifies the specific regions in the pseudo-observation space responsible for the rejection, thereby giving a more interpretable description of model lack-of-fit.

Several directions remain for future work. First, the current mixture model selection procedure relies on parameters estimated by $k$-means, and it would be useful to develop a version that explicitly accounts for parameter estimation uncertainty. Additionally, the Monte Carlo version of this GOF test could potentially be sharpened by replacing Clopper--Pearson intervals with less conservative simultaneous intervals while retaining rigorous error control. It would also be interesting to study the power of the Beta-tree GOF test against structured local alternatives and to compare its localization properties more systematically with existing global and local goodness-of-fit procedures. Finally, it would be interesting to explore alternative definitions of the GOF score that incorporate not only the number of flagged regions, but also the magnitude and location of the discrepancies between the fitted null distribution and the data.

\printbibliography

\end{document}